\documentclass[11pt, letterpaper]{amsart}
\usepackage{graphicx, amssymb}

\newtheorem{theorem}{Theorem}[section]
\newtheorem{assumption}{Assumption}[section]

\newtheorem{corollary}[theorem]{Corollary}

\newtheorem{definition}[theorem]{Definition}

\newtheorem{lemma}[theorem]{Lemma}

\newtheorem{proposition}[theorem]{Proposition}
\newtheorem{remark}[theorem]{Remark}

\numberwithin{theorem}{section}

\newcommand{\mc}{\mathcal}
\newcommand{\ra}{\rightarrow}
\newcommand{\mb}{\mathbb}
\DeclareMathOperator{\es}{ess\,sup \ }

\title[]{Stability of exponential utility maximization with respect to market perturbations}\thanks{We would like to thank the referees for their feedback, which has helped us improve our paper greatly.}
\author[]{Erhan Bayraktar}\thanks{The authors are supported in part by the National Science Foundation under an applied mathematics research grant and a Career grant, DMS-0906257 and DMS-0955463, respectively, in part by the Susan M. Smith Professorship, and in part by the NDSEG Fellowship Program of the Department of Defense.} 
\address[Erhan Bayraktar]{Department of Mathematics, University of Michigan, 530 Church Street, Ann Arbor, MI 48104, USA}
\email{erhan@umich.edu}
\author[]{Ross Kravitz}
\address[Ross Kravitz]{Department of Mathematics, University of Michigan, 530 Church Street, Ann Arbor, MI 48104, USA}
\email{ross.kravitz@gmail.com}

\date{December 11, 2012}
\numberwithin{equation}{section}
\begin{document}
\begin{abstract} We investigate the continuity of expected exponential utility maximization with respect to perturbation of the Sharpe ratio of markets.  By focusing only on continuity, we impose weaker regularity conditions than those found in the literature.  Specifically, we require, in addition to the $V$-compactness hypothesis of \cite{MR2438002}, a local $bmo$ hypothesis, a condition which is essentially implicit in the setting of \cite{MR2438002}.  For markets of the form $S = M + \int \lambda d\langle M \rangle$, these conditions are simultaneously implied by the existence of a uniform bound on the norm of $\lambda \cdot M$ in a suitable $bmo$ space.
\end{abstract}

\maketitle

\section{Introduction}\label{exp_int}

In this paper we provide stability results for the problem of maximizing expected exponential utility.  We give conditions under which convergence of markets implies the convergence of optimal terminal wealths as well as their expected utility.  Specifically, for markets of the form $S = M + \int \lambda d \langle M \rangle$, our regularity condition consists of two complementary hypotheses: the first, the familiar $V$-compactness assumption of \cite{MR2438002}, is used to establish lower semi-continuity, while the second, a new condition related to a local $bmo$ hypothesis, is used to establish upper semi-continuity.  Both the $V$-compactness and local $bmo$ conditions originally arose as consequences of our original regularity condition, a uniform bound on the $bmo_2$ norm of $\lambda \cdot M$.  This type of hypothesis is a natural one in mathematical finance and has, for example, appeared in \cite{MR1305680} and \cite{MR1655145}, where it was used in connection with establishing closedness properties of the space of attainable terminal wealths.  In the current setting, the $bmo$ hypothesis allows us to find wealth processes which are simultaneously near optimal and bounded from below.  This is useful because the optimal wealth process is in general unbounded, meaning that it may go arbitrarily far into the red.  With this approximation result in hand, we may use the stability results of \cite{MR2438002} for utility functions on $\mb{R}_+$ to obtain convergence under $bmo$ regularity.  From there, we can prove our most general continuity result, which builds upon the $bmo$ arguments to establish upper semi-continuity.

In comparison with \cite{MR2438002}, dealing with the stability problem for utility functions on $\mb{R}_+$, our regularity assumption is of course stronger than the notion of $V$-compactness alone, since we impose the additional local $bmo$ condition.  We will show, however, that this condition is not especially stringent, and that on some level it is already implicit in the setting of \cite{MR2438002}: indeed, the basic purpose of the assumption is to guarantee that over all markets, the optimal expected utility $E \left[ U \left( \widehat{X}^n_T \right) \right]$ may be uniformly approximated  by payoffs of the form $E \left[ U \left(\widehat{X}^{(n,k)}_T \right) \right]$, with the processes $\widehat{X}^{(n,k)}$ satisfying $\underset{0 \leq t \leq T}{\sup} U^- \left( \widehat{X}^{(n,k)}_t \right) \in L^\infty$.  For utility functions defined on $\mb{R}_+$, this property is guaranteed as soon as there is $V$-compactness.

In comparison with two other extant stability results in the literature, we see that our regularity hypothesis is weaker than in either of those papers, although they provide additional convergence results that are beyond the scope of this article.  In \cite{FreiBSDE}, the stability of quadratic BSDE's is studied with respect to, among other things, perturbation of the driver.  From the natural connection between this class of BSDE's and exponential utility maximization (see \cite{MR2465489}), the results from \cite{FreiBSDE} allow one to recover stability results about exponential utility maximization, but only under the more restrictive assumption of a uniform bound on $\lambda \cdot M$ in the Hardy Space $H^\infty$, i.e. $||\lambda \cdot M||_{H^\infty} \triangleq ||\lambda^2 \cdot \langle M \rangle_T||_{L^\infty}$.  Additionally, it is assumed that the filtration is continuous.%, so that all local martingales have continuous paths.  With these assumptions, Frei proves a general stability result for quadratic BSDE's that can be applied to other problems like stability with respect to trading constraints.

In \cite{Zit09c}, very strong convergence results are obtained in a narrow class of utility maximization problems, with equilibrium problems in mind.  In order to use PDE methods, the setting is exclusively Markov, and the assumptions on market convergence are quite stringent: given a sequence $(\lambda^n)_{n=1,\ldots,\infty}$ of drift parameters, essentially $\lambda^n(t) \ra \lambda^\infty(t)$ in $L^\infty([0,T])$.  These strong hypotheses are necessary to deduce quantitative estimates about the stability of exponential utility maximization.

Here, we take a different approach.  We consider simply the continuity of exponential utility maximization in a general filtration, and are interested in finding minimal regularity conditions under which continuity will hold true.  The outline of the paper is as follows.  In Section \ref{exp_setup}, we provide the necessary background definitions to state our main result.  In Section \ref{exp_bmo}, we present some preliminaries on the theory of bmo martingales.  In Section \ref{exp_approx}, we apply this theory to give a proof of an intermediate result.  In Section \ref{exp_prf}, we prove the main results of the paper, using the results of Section \ref{exp_approx} to establish upper semi-continuity.  Finally, in Section \ref{exp_ass}, we discuss our second assumption in the context of \cite{MR2438002} and discuss its economic significance, in connection with the opportunity process of \cite{MR2892960} and \cite{MR2721695}; the necessity of the first assumption is also addressed.  We close with two appendices, \ref{exp_app1} and \ref{exp_app2}, which contain auxiliary results.

\section{Setup and Main Results}\label{exp_setup}

Let $(\Omega,\mc{F},P,(\mc{F}_t)_{t \in [0,T]})$ be a filtered probability space satisfying the usual conditions.  We assume that $\mc{F}_T = \mc{F}$.  Let $M$ be a continuous local martingale, and let 
\[\Lambda \triangleq \left \{ \lambda : \lambda \text{ is a predictable process satisfying } \int_0^T \lambda_u^2 d \langle M \rangle_u < \infty \ \right \}.
\]  
For $\lambda \in \Lambda$, define 
\begin{equation}\label{eq1}
S_t^\lambda \triangleq M_t + \int_0^t \lambda_u d \langle M \rangle_u,
\end{equation}
where $ \langle M \rangle = \left(\langle M \rangle_t\right)_{t \in [0,T]}$ denotes the quadratic variation of the local martingale $M$.
Along with a num\'{e}raire bond, identically equal to $1$, each $S^\lambda$ defines a stock market, in which $S^\lambda$ is interpreted as the discounted price of a tradeable asset.  
%We assume that every market under consideration satisfies NFLVR, which is equivalent to the existence of an equivalent local martingale measure.  It is proven in \cite{MR1384360} that all continuous market models satisfying NFLVR have the specific form of \eqref{eq1}.

We let $S^n \triangleq M + \int \lambda^n \ d\langle M \rangle$, $n=1,\ldots,\infty$, describe a sequence of markets, and 
\[
\begin{split}
Z^n  & \triangleq \mc{E}(-\lambda^n \cdot M) \\
& = \exp \left( - \int_0^\cdot \lambda^n dM - \frac{1}{2} \int_0^\cdot (\lambda^n)^2 d \langle M \rangle \right)
\end{split}
\]
 is the $nth$ minimal martingale measure.

In the exponential utility maximization problem, an agent with utility function $U(x) \triangleq -\exp(-x)$ seeks to maximize $E \left[U(x + X_T) \right]$ over a set of admissible wealth processes $X$ that start from initial capital zero.  We set $V(y) \triangleq y \log y - y$ for $y>0$, so that $V$ is the convex dual of $U$.  To define our regularity assumptions, we need the notion of $bmo$ martingales.

\begin{definition} Let $1 \leq p<\infty$.  A not necessarily continuous martingale $R$ is in $bmo_p$, with $||R||_{bmo_p} = r$, if there is a minimal constant $r$ such that
\[ E \left[ | R_T - R_{\tau-}|^p \ | \ \mc{F}_\tau \right]^{\frac{1}{p}} \leq r,
\]
for all stopping times $\tau$ taking values in $[0,T]$.  We will occasionally abbreviate $bmo_1$ to $bmo$.
\end{definition}

For $p=2$, if $||R||_{bmo_2} < \infty$, then $||R||_{bmo_2}$ also has the representation 
\[
\underset{\tau}{\es} \  E \left[ \langle R \rangle_T - \langle R \rangle_{\tau-} \ | \ \mc{F}_\tau \right]^{\frac{1}{2}}.
\]
The equivalence of this representation is derived from considering the martingale $R^2 - \langle R \rangle$. 

Now we can state our two-pronged regularity assumption on a sequence of markets:

\begin{assumption}\label{regass1}[Regularity Assumption 1: $V$-compactness] The set \\ $\left \{ V \left( Z^n_T \right) : n \in \mathbb{N} \right \}$ is uniformly integrable.
\end{assumption}

\begin{assumption}\label{regass2}[Regularity Assumption 2] There exists a sequence of stopping times $(\tau_j) \uparrow T$ such that $\underset{n}{\sup} \ ||(\lambda^n \cdot M)^{\tau_j}||_{bmo_2} < \infty$ for each $j$.
\end{assumption}

We continue on with our description of the utility maximization problem.  In comparison to utilities on $\mb{R}_+$, defining the right notion of admissibility is more complicated when $U$ is finite-valued over the whole real line.  We state here the most common definition of admissibility at this level of generality, for which we refer to \cite{MR2014244}.  Let $\mc{M}^n$ denote the set of equivalent local martingale measures for $S^n$.  

%\begin{remark}\label{martremark} Under the assumption that $V(Z^n_T) \in L^1$, immediately implied by $V$-compactness, it is known (see, e.g. Ex 1.16, p. 58 of \cite{STMAZ.00635670}) that $Z^n \in \mc{H}^1$, which means that $Z^n \in \mc{M}^n$, the set of equivalent local martingale measures for $S^n$, as opposed to $Z^n$ being a strict local martingale.  In particular, $\mc{M}^n$ will be nonempty.  
%\end{remark}

\begin{definition}\label{admdef} For any $n$, let $H$ be predictable and $S^n$-integrable.  We say that $H \cdot S^n \in \mc{A}^n$ if $H \cdot S^n$ is a $\mb{Q}$-martingale for every $\mb{Q} \in \mc{M}^n$ with finite entropy, that is, $E \left[V \left( \frac{d \mb{Q}}{dP} \right) \right] < \infty$.
\end{definition}

The primal value function $u^n$, $n=1,\ldots,\infty$, is defined as
\[
u^n(x) \triangleq \underset{X \in \mc{A}^n}{\sup} E \left[ U(x + X_T) \right], x \in \mb{R}.
\]

\noindent In the stability problem for utility maximization, we seek assumptions on the processes $Z^n$ that ensure the convergence of $u^n(\cdot)$ towards $u^\infty(\cdot)$.  We can now state the main results of the paper:

\begin{theorem}\label{fin.mainthm}  Suppose that $Z_T^n \ra Z_T^\infty$ in probability, $Z^\infty$ is a martingale, and that Assumptions \ref{regass1} and \ref{regass2} are satisfied.  Then $u^n(\cdot) \ra u^\infty(\cdot)$ pointwise, hence locally uniformly.
\end{theorem}

\begin{theorem}\label{fin.mainthm2} Suppose that $Z_T^n \ra Z_T^\infty$ in probability, $Z^\infty$ is a martingale, and that Assumptions \ref{regass1} and \ref{regass2} are satisfied. Then for all $x$ the optimal terminal wealths $\widehat{X}^n_T(x)$ converge to $\widehat{X}^\infty_T(x)$ in probability as $n \ra \infty$.
\end{theorem}

A crucial intermediate step in establishing these theorems lies in first establishing them under a stronger $bmo$-type hypothesis.  This is the main intermediate theorem: we remark that under these assumptions, $Z^\infty$ is automatically a martingale.

\begin{theorem}\label{mainthm}  Suppose that $Z_T^n \ra Z_T^\infty$ in probability and that $\underset{n}{\sup} \  ||\lambda^n \cdot M||_{bmo_2} < \infty$.  Then $u^n(\cdot) \ra u^\infty(\cdot)$ pointwise, hence locally uniformly.
\end{theorem}

%\begin{theorem}\label{mainthm2} Suppose that $Z_T^n \ra Z_T^\infty$ in probability and that $\underset{n}{\sup} \  ||\lambda^n \cdot M||_{bmo_2} < %\infty$.  Then for all $x$ the optimal terminal wealths $\widehat{X}^n_T(x)$ converge to $\widehat{X}^\infty_T(x)$ in probability as $n \ra \infty$.
%\end{theorem}

\section{BMO Preliminaries}\label{exp_bmo}

\begin{definition} A positive martingale $Y$ satisfies the Reverse H\"{o}lder Inequality $\mc{R}_p(\mb{P})$ for $p>1$ with constant $K_p$ and with respect to the measure $\mb{P}$, if there exists minimal $K_p$ such that
\[
E^{\mb{P}} \left[ \frac{Y^p_T}{Y^p_\tau} \ \bigg| \ \mc{F}_\tau \right] \leq K_p
\]
for all stopping times $\tau$ in $[0,T]$.
\end{definition}

The following lemma is found in the appendix of \cite{MR1920099}, and originally in Propositions $5$ and $6$ of \cite{MR544804}.

\begin{lemma}\label{lemmabmoholder}  Suppose that the collection $\left(\lambda^n \cdot M \right)_{n \geq 1}$ is bounded in the $bmo_2$ norm.  Then for some $p>1$ which depends only on this uniform bound, the collection $\left(Z^n = \mc{E}(\lambda^n \cdot M) \right)_{n\geq 1}$ satisfies $\mc{R}_p(\mb{P})$ with, respectively, uniformly bounded constants $C^n_p$.
\end{lemma}

\begin{definition}\label{revholddef} A positive martingale $Y$ satisfies $\mc{R}_{LLogL}$ with constant $K_{LLogL}$ if there exists minimal $K_{LLogL}$ such that
\[ E \left[ V \left(\frac{Y_T}{Y_\tau} \right) \ \bigg| \ \mc{F}_\tau \right] \leq K_{LLogL}
\]
for all stopping times $\tau$ in $[0,T]$.
\end{definition}

\begin{definition}  A positive c\`{a}dl\`{a}g process $Y$ satisfies condition (S) if there exist constants $0 < c \leq 1 \leq C$ such that $ cY_- \leq Y \leq C Y_-$.
\end{definition}

%For a stopping time $\tau$, define $^\tau Z_T = \frac{Z_T}{Z_\tau}$.  
The following proposition is mostly in the literature:

\begin{proposition}\label{exp_prop1}  Let $R$ be a martingale such that $Y = \mc{E}(R)$ is a strictly positive martingale.  Then $R \in bmo_2$ and there exists $h>0$ such that $\Delta R \geq h-1$ if and only if $Y$ satisfies $\mc{R}_{LLogL}$ and condition (S).  
%Additionally the $bmo_2$ norm of $M$ can be bounded by an (increasing) function of the $\mc{R}_{LLogL}$ and condition (S) constant $C$ of $Z$
The constants $K_{LLogL}$ and $C$ of $Y$ can be bounded as a function of $||R||_{bmo_2}$.  
\end{proposition}

\begin{proof}  In the $(\Leftarrow)$ direction, Lemma $2.2$ of \cite{MR1920099} establishes that $R \in bmo_2$.  Now $dY = Y_- dR$ and $\Delta Y = Y_- \Delta R$.  By the first inequality of condition (S), $(c -1 )Y_- \leq Y - Y_- = Y_- \Delta R$, implying that $\Delta R \geq c - 1 > -1$.

  Now the $(\Rightarrow)$ direction.  Since $R$ is in $bmo_2$ it is locally bounded; indeed, for $n \in \mb{N}$, let $\tau_n = \inf \{ t : \Delta R_t \geq n \} \wedge T$, and let $r \triangleq ||R||_{bmo_2}$.  Then $||R^{\tau_n}||_{bmo_2} \leq r$, so that 
\[
\begin{split}
(\Delta R_{\tau_n})^2 & = \Delta \langle R \rangle_{\tau_n} \\
& = E \left[ \langle R \rangle_{\tau_n} - \langle R \rangle_{\tau_n-} \ | \ \mc{F}_{\tau_n} \right] \\
& \leq r,
\end{split}
\]
so that the jumps of $R$ are bounded in magnitude by $\sqrt{r}$.  This implies that $R$ is locally bounded.  Then $\Delta Y = Y_- \Delta R \leq \sqrt{r} Y_-$.  Hence $Y \leq Y_- + \sqrt{r}Y_-$.  Additionally, $\Delta R \geq h -1$ implies that 
\[
\begin{split}
Y - Y_- & = \Delta Y \\
& Y_- \Delta R \\
& \geq Y_- (h-1),
\end{split}
\]
so $Y \geq hY_-$.  This establishes condition (S), with $C = 1 + \sqrt{r}$, which is bounded as a function of $||R||_{bmo_2}$.
  
  By Lemma \ref{lemmabmoholder}, $Y$ satisfies the reverse H\"{o}lder inequality for some $p>1$.  Since $x \log x \leq K' x^p$ for some constant $K'$, it follows that $Y$ satisfies $\mc{R}_{LLogL}$.  Additionally, it is evident that Lemma \ref{lemmabmoholder} also implies that $Y$ satisfies $\mc{R}_{LLogL}$ with constant $K_{LLogL}$ only depending on $||R||_{bmo_2}$.
\end{proof}

\begin{definition}
For each market $n$, let $\widehat{Z}^n $ be the minimal entropy martingale measure.  Its existence and uniqueness is established in Theorem $2.2$ of \cite{MR1865021}.
\end{definition}

The next lemma is precisely Lemma $3.1$ of \cite{MR1891730}.  We give a proof for the reader's convenience.  
%Another proof can be given using the optimality property associated with the dual value function.  For this proof to be made rigorous some technical properties of the value function must be proven, but we state a simple and intuitive heuristic proof below, which is given as Corollary $2.1$ in \cite{MR1994915}.
\begin{lemma}\label{exp_lemma1} For any $n$, if $Z^n$ satisfies $\mc{R}_{LLogL}$, with constant $K$, then $\widehat{Z}^n$ satisfies $\mc{R}_{LLogL}$ with a constant less than or equal $K$.
\end{lemma}

\begin{proof}
By hypothesis, $E \left[ \frac{Z^n_T}{Z^n_\tau} \log \frac{Z^n_T}{Z^n_\tau} \Big| \mc{F}_\tau \right] \leq K$ for all stopping times $\tau$ less than than or equal to $T$.  Suppose that $\widehat{Z}^n$ does not satisfy $\mc{R}_{LLogL}$ with a constant less than or equal $K$.  Then there exists $\epsilon > 0$, a stopping time $\sigma$ less than or equal to $T$, and a set $A \in \mc{F}_\sigma$ with $P(A) > 0$ such that
\[
E \left[ \frac{\widehat{Z}^n_T}{\widehat{Z}^n_\sigma} \log \frac{\widehat{Z}^n_T}{\widehat{Z}^n_\sigma} \Big| \mc{F}_\sigma \right] \geq K + \epsilon
\]
on the set $A$.
Let $\tilde{Z}^n_t \triangleq 1_{\{t < \sigma\}} \widehat{Z}^n_t + 1_{\{t \geq \sigma\}} \left( 1_A \frac{Z^n_t}{Z^n_\sigma} \widehat{Z}^n_\sigma + 1_{A^c} \widehat{Z}^n_t \right)$ for $t \in [0,T]$.  Then $\tilde{Z}^n$ is the density process of an element of $\mc{M}^n$ and satisfies $\tilde{Z}^n_T = 1_A \widehat{Z}^n_\sigma \frac{Z^n_T}{Z^n_\sigma} + 1_{A^c} \widehat{Z}^n_T$.  
Thus,
\[
\tilde{Z}^n_T \log \tilde{Z}^n_T = 1_{A^c} \widehat{Z}^n_T \log \widehat{Z}^n_T + 1_A \left( \widehat{Z}^n_\sigma \frac{Z^n_T}{Z^n_\sigma} \log \frac{Z^n_T}{Z^n_\sigma} + \widehat{Z}^n_\sigma \frac{Z^n_T}{Z^n_\sigma} \log \widehat{Z}^n_\sigma \right).
\]
Therefore,

\begin{eqnarray*}
\lefteqn{E \left[ \widetilde{Z}^n_T \log \widetilde{Z}^n_T | \mc{F}_\sigma \right] - E \left[\widehat{Z}^n_T \log \widehat{Z}^n_T | \mc{F}_\sigma \right]} \\
&& =1_A \left( \widehat{Z}^n_\sigma E \left[ \frac{Z^n_T}{Z^n_\sigma} \log \frac{Z^n_T}{Z^n_\sigma} \big| \mc{F}_\sigma \right] + \widehat{Z}^n_\sigma \log \widehat{Z}^n_\sigma - E \left[ \widehat{Z}^n_T \log \widehat{Z}^n_T | \mc{F}_\sigma \right] \right) \\
&& =1_A \widehat{Z}^n_\sigma \left( E \left[ \frac{Z^n_T}{Z^n_\sigma} \log \frac{Z^n_T}{Z^n_\sigma} \big| \mc{F}_\sigma \right] - E \left[ \frac{\widehat{Z}^n_T}{\widehat{Z}^n_\sigma} \log \frac{\widehat{Z}^n_T}{\widehat{Z}^n_\sigma} \big| \mc{F}_\sigma \right] \right) \\
&& \leq -\epsilon 1_A \widehat{Z}^n_\sigma. 
\end{eqnarray*}
Taking expectations, this contradicts the fact that $\widehat{Z}^n_T$ has minimal entropy.
\end{proof}

We now show that the $bmo_2$ hypothesis of (\ref{mainthm}) implies the $V$-compactness condition of Assumption \ref{regass1}, which plays a prominent role in \cite{MR2438002}.  The next proposition is proven for continuous martingales in \cite{MR1299529}.

\begin{proposition}\label{exp_prop2} Suppose $\underset{n}{\sup} \ ||\lambda^n \cdot M||_{bmo_2} < \infty$.  Then there exists $p>1 $ such that $\underset{n}{\sup} \  E \left[ (Z^n_T)^p \right] < \infty$.
\end{proposition}

\begin{proof}  By the conditional form of  Jensen's inequality, the norm $||\cdot||_{bmo_1} \leq ||\cdot||_{bmo_2}$.  Let $R$ be an arbitrary element of $bmo_2$, and let $n(R) = 2||R||_{bmo_1} + ||R||^2_{bmo_2}$.  Without loss of generality, we assume that $||R||_{bmo_2} > 0$, and show that the $L^p$ norm of $\mc{E}(R)_T$ has an upper bound that only depends on $n(R)$ for some $p>1$.

Let $\delta =\exp(-pn(R)) < 1$ (so $\log 1/\delta = pn(R)$), and let $\tau = \inf \{t : \mc{E}(R)_t > \lambda \}$ for $\lambda >1$.  Considering time $\tau-$ instead of $\tau$ and arguing as in \cite{MR1299529}, we obtain the inequality $P(\mc{E}(R)_T/\mc{E}(R)_{\tau-} \geq \delta \ | \ \mc{F}_\tau) \geq 1 - \frac{1}{2p}$; indeed,
\begin{eqnarray*}
\lefteqn{P \left( \mc{E}(R)_T/\mc{E}(R)_{\tau-} < \delta \ | \ \mc{F}_\tau \right)} \\
&&= P \left( 1/\delta < \mc{E}(R)_{\tau-}/\mc{E}(R)_T \ | \ \mc{F}_\tau \right) \\
&& = P \left( pn(R) < R_{\tau-} - R_T + \frac{1}{2}(\langle R \rangle_T - \langle R \rangle_{\tau-}) \ | \ \mc{F}_\tau \right) \\
&& \leq \frac{1}{2pn(R)}  E \left[ 2|R_T - R_{\tau-}| + (\langle R \rangle_T - \langle R \rangle_{\tau-}) \ | \ \mc{F}_\tau \right] \\
&& \leq \frac{n(R)}{2pn(R)} \\
&& = \frac{1}{2p},
\end{eqnarray*}
\noindent with the first inequality following from Markov's inequality.  This implies that \\ $P \left( \mc{E}(R)_T/\mc{E}(R)_{\tau-} \geq \delta \ | \ \mc{F}_\tau \right) \geq 1 - \frac{1}{2p}$.

By Proposition \ref{exp_prop1}, $\mc{E}(R)$ satisfies the upper bound of condition (S) with a constant $C$ whose size is controlled by $n(R)$, and we have $\mc{E}(R)_{\tau-} \geq \frac{1}{C}\mc{E}(R)_{\tau} \geq \frac{1}{C}\lambda$ on $\{\tau < \infty \}$.  This yields $P \left( \mc{E}(R)_T \geq \frac{\delta \lambda}{C} \ | \ \mc{F}_\tau \right) \geq \frac{2p-1}{2p}1_{\{\tau < \infty\}}$.
Thus,
\begin{eqnarray*} 
\lefteqn{E \left[\mc{E}(R)_T 1_{\{\mc{E}(R)_T > \lambda \}} \right]} \\
&& \leq E \left[ \mc{E}(R)_T 1_{\{\tau < \infty \}} \right] \\
&& = E \left[ \mc{E}(R)_\tau 1_{\{\tau < \infty \}} \right] \\
&& \leq E \left[ C \mc{E}(R)_{\tau-} 1_{\{\tau < \infty \}} \right] \\
&& \leq C \lambda P(\tau < \infty) \\
&& \leq \frac{2C \lambda p}{2p - 1} P \left( \mc{E}(R)_T \geq \frac{\delta \lambda}{C} \right),
\end{eqnarray*}
where the equality above follows from the optional sampling theorem.

Take the inequality $E \left[ \mc{E}(R)_T 1_{\{\mc{E}(R)_T > \lambda \}} \right] \leq \frac{2C \lambda p}{2p - 1} P \left( \mc{E}(R)_T \geq \frac{\delta \lambda}{C} \right)$, multiply both sides by $(p-1)\lambda^{p-2}$ and integrate with respect to $\lambda$ from $1$ to $\infty$:
\begin{eqnarray}\label{some}
\lefteqn{\int_1^\infty (p-1) \lambda^{p-2} E \left[ \mc{E}(R)_T 1_{\{\mc{E}(R)_T > \lambda \}} \right] d \lambda} \\
&&\leq \int_1^\infty (p-1) \lambda^{p-2} \frac{2C \lambda p}{2p - 1} P \left( \mc{E}(R)_T \geq \frac{\delta \lambda}{C} \right) d\lambda. \label{lasteq}
\end{eqnarray}
Applying Fubini's Theorem to the left hand side \eqref{some}, we get
\begin{eqnarray*}
\lefteqn{\int_1^\infty (p-1) \lambda^{p-2} E \left[ \mc{E}(R)_T 1_{\{\mc{E}(R)_T > \lambda \}} \right] d \lambda} \\
&& =E \left[\int_1^\infty (p-1)\lambda^{p-2} \mc{E}(R)_T 1_{\{\mc{E}(R)_T > \lambda \}} d \lambda \right] \\
&& = E \left[ \mc{E}(R)_T \int_1^\infty (p-1)\lambda^{p-2}  1_{\{\mc{E}(R)_T > \lambda \}} d \lambda \right] \\
&& = E \left[ \mc{E}(R)_T 1_{\{\mc{E}(R)_T > 1 \}} \int_1^{\mc{E}(R)_T} (p-1)\lambda^{p-2} d \lambda \right] \\
&& = E \left[ \mc{E}(R)_T \left(\mc{E}(R)_T^{p-1} - 1 \right) 1_{\{\mc{E}(R)_T > 1 \}} \right].
\end{eqnarray*}

After a similar computation for the right hand side \eqref{lasteq}, this yields
\begin{eqnarray*}
\lefteqn{E \left[ (\mc{E}(R)_T^p - \mc{E}(R)_T) 1_{\{\mc{E}(R)_T > 1\}} \right]} \\
&& \leq \frac{2C(p-1)}{2p-1} E \left[ \left(\left(\frac{C}{\delta}\mc{E}(R)_T \right)^p - 1 \right) 1_{\{\mc{E}(R)_T > \frac{\delta}{C}\}} \right].
\end{eqnarray*}

Grouping the terms with $\mc{E}(R)^p_T$ together on the left hand side, we obtain
\begin{eqnarray*}
\lefteqn{\left(1 - \frac{2C(p-1)}{2p -1}\frac{C^p}{\delta^p} \right) E \left[ \mc{E}(R)_T^p 1_{\{\mc{E}(R)_T > 1 \}} \right]} \\ 
&& \leq E \left[ \mc{E}(R)_T \right] - \frac{2C(p-1)}{2p-1} E \left[ 1_{\{\mc{E}(R)_T > \frac{\delta}{C}\}} \right] \\
&& \leq 1,
\end{eqnarray*}
for any $p>1$.  Hence, by choosing $p$ close enough to $1$ so that $\frac{2C(p-1)}{2p-1}\frac{C^p}{\delta^p} < 1$, we establish an upper bound for $E[\mc{E}(R)_T^p]$ which depends only on $n(R)$.  Note that the choice of $C$ depends on $n(R)$.
\end{proof}

\begin{corollary}\label{exp_cor1} Suppose that $\underset{n}{\sup} \ ||\lambda^n \cdot M ||_{bmo_2} < \infty$.  Then $\{V(Z_T^n) : n \in \mb{N} \}$ is uniformly integrable.
\end{corollary}

\begin{proof}  By Proposition \ref{exp_prop2}, $\underset{n}{\sup} \ E[(Z_T^n)^p] < \infty$ for some $p>1$.  As $x^{\widetilde{p}}/V(x) \ra \infty$ as $x \ra \infty$, for any $\widetilde{p}>1$, the claim follows from the de la Vall\'{e}e-Poussin criterion.
\end{proof}

We make one last digression to the theory of bmo martingales.  Specifically, we need the bmo theory of weighted norm inequalities.  The following theorem is stated as Theorem $2.16$ of \cite{norm} without mentioning that the constant $C_p$ in \eqref{normeq} can be chosen as the same constant associated with the reverse H\"{o}lder inequality.  For this fact, we refer to Proposition $2$ of \cite{MR544804}.  For a c\`{a}dl\`{a}g process $Y$, let $Y^* \triangleq \underset{t \in [0,T]}{\sup} \ |Y_t| \in \mc{F}_T$. 

\begin{proposition}\label{exp_prop4} Let $Y=\mc{E}(R)$ be a continuous martingale and $\frac{d\mb{Q}}{dP} = Y_T$.  Then if $Y$ satisfies $\mathcal{R}_p(P)$ with constant $C_p$, then for each $\mb{Q}$-martingale $X$ and $q = \frac{p}{p-1}$,
\begin{equation}\label{normeq}
\lambda^q P \left( X^* > \lambda \right) \leq C_p E \left[ |X_T|^q \right].
\end{equation}
\end{proposition}

\section{Approximation of Optimal Wealth}\label{exp_approx}

In \cite{MR1865021}, $U$ is approximated by auxiliary utility functions defined on a half axis.  For $k \in \mb{N}$, we define utility functions $U^{(k)}$ as follows: $U^{(k)} = U$ on $[-k,\infty)$, $U(x) \geq U^{(k)}(x) > -\infty$ for $x > -k - 1$, and $\underset{x \downarrow -k - 1}{\lim} U^{(k)}(x) = -\infty$.  Each $U^{(k)}$ is assumed $C^1$, concave, satisfying the Inada conditions, and having reasonable asymptotic elasticity.  For details on these assumptions, see \cite{MR1865021}.  $V^{(k)}$ is the convex conjugate of $U^{(k)}$.  Since $U^{(k)} \leq U$, $V^{(k)} \leq V$.  

For $n=1,\ldots,\infty$, $v^n$ is the dual value function associated to $V$ in market number $n$: 
\begin{equation}\label{dualvaluedef}
v^n(y) \triangleq \underset{\mb{Q} \in \mc{M}^n}{\inf} E \left[ V \left(y \frac{d \mb{Q}}{dP} \right) \right], \  y>0.
\end{equation}
For $n=1,\ldots,\infty$ and $k \in \mb{N}$, $v^{(n,k)}$ is the dual value function associated to $V^{(k)}$ in market number $n$:
\[
v^{(n,k)}(y) \triangleq \underset{Y \in \mc{Y}^n}{\inf} E \left[ V^{(k)}(yY_T) \right], \  y>0,
\]
where $\mc{Y}^n$ is the set of supermartingale deflators for $S^n$:
\begin{definition}  $\mc{Y}^n$ is the set of c\`{a}dl\`{a}g processes $Y$ such that $Y_0=1$ and $Y(H \cdot S^n)$ is a supermartingale whenever $H$ is predictable, $S^n$-integrable, such that $H \cdot S^n$ is bounded from below by a constant.
\end{definition}

Let $\mc{A}_b^n$ be the set of wealth processes $H \cdot S^n$ where $H$ is predictable and $S^n$-integrable, and $H \cdot S^n$ is bounded from below by a constant.  The value functions $u^{(n,k)}$ are defined as follows: 
\[
u^{(n,k)}(x) \triangleq \underset{X \in \mc{A}_b^n}{\sup} E \left[ U^{(k)}(x + X_T) \right], \  x > -k-1.
\]

By a shift on the real line (see \cite{MR1865021}), one can identify the value functions $v^{(n,k)},u^{(n,k)}$ with an equivalent optimization problem which uses a utility function $\widetilde{U}^{(k)}$ defined on $\mb{R}_+$.  We copy verbatim this procedure here.

Let $\widetilde{U}^{(k)}(x) \triangleq U^{(k)}(x - (k+1))$, which is finitely valued for $x>0$.  Then $\widetilde{U}^{(k)}$ is a utility function of the type encountered in \cite{MR1722287}, and so there is a unique optimal solution $\overline{X}^{(n,k)}(x) = x + H^{(n,k)} \cdot S^n$ to the optimization problem 
\[
\widetilde{u}^{(n,k)}(x) \triangleq \underset{X \in \mc{A}^n_b}{\sup} \  E \left[ \widetilde{U}^{(k)}(X_T) \right], \ x > -k - 1.
\]
Then, for $x > -k-1$, $\widehat{X}^{(n,k)}(x) \triangleq \overline{X}^{(n,k)}(x + k + 1) - (k +1)$ is the optimal solution to the optimization problem
\[
u^{(n,k)}(x) = \underset{X \in \mc{A}^n_b}{\sup} \  E \left[ U^{(k)}(x + X_T) \right], \ x>0. 
\]
It follows that $u^{(n,k)}(x) = \widetilde{u}^{(n,k)}(x+k+1)$ for $x > -k-1$.  Let $\widetilde{V}^{(k)}$ be the convex conjugate of $\widetilde{U}^{(k)}$.  Then the convex conjugate $\widetilde{v}^{(n,k)}$ of $\widetilde{u}^{(n,k)}$ has the form
\[
\begin{split}
\widetilde{v}^{(n,k)}(y) & = \underset{Y \in \mc{Y}^n}{\inf} E \left[ \widetilde{V}^{(k)}(yY_T) \right] \\
& = E \left[ \widetilde{V}^{(k)} \left( y\widetilde{Y}_T^{(n,k)} \right) \right], \ y>0;
\end{split}
\]
Here, $\widetilde{Y}^{(n,k)} = \widetilde{Y}^{(n,k)}(y)$ is the dual minimizer, which in general depends on $y$; the existence of such minimizers is established in \cite{MR1722287}.  We also have 
\begin{equation}\label{Vtilde}
V^{(k)}(y) = \widetilde{V}^{(k)}(y) + (k+1)y
\end{equation}
\noindent and $v^{(n,k)}(y) = \widetilde{v}^{(n,k)}(y) + (k+1)y$.  
The main result of \cite{MR2438002} implies that for each $k$, $\underset{n \ra \infty}{\lim} \widetilde{u}^{(n,k)} = \widetilde{u}^{(\infty,k)}$ under the $\widetilde{V}^k$-compactness condition: $\{ \widetilde{V}^k(Z_T^n) : n \in \mb{N} \}$ is uniformly integrable.

\begin{lemma}\label{exp_lemma3}  Suppose that $Z_T^n \ra Z_T^\infty$ in probability and $\{Z_T^n : n \in \mb{N}\}$ is $V$-compact.  Then for each $k$,  $\underset{n \ra \infty}{\lim} u^{(n,k)}(x) = u^{(\infty,k)}(x)$.
\end{lemma}

\begin{proof}  For each $k$, $V^{(k)} \leq V$ and $V^{(k)}$ is bounded from below, so $\{V^{(k)}(Z_T^n) : n \in \mb{N}\}$ is uniformly integrable.  Since $V(x)/x \ra \infty$ as $ x \ra \infty$, it is also true that $\{Z_T^n : n \in \mb{N}\}$ is uniformly integrable.  Given the form of $\widetilde{V}^{(k)}$ in \eqref{Vtilde}, it now follows that $\{\widetilde{V}^{(k)}(Z_T^n) : n \in \mb{N}\}$ is uniformly integrable.  Hence the main theorem of \cite{MR2438002} implies that $\widetilde{u}^{(n,k)}(x) \ra \widetilde{u}^{(\infty,k)}(x)$ as $n \ra \infty$.  It immediately follows that $u^{(n,k)}(x) \ra u^{(\infty,k)}(x)$ as $n \ra \infty$.
\end{proof}

\begin{lemma}\label{exp_lemma4}  Suppose that $v^*(y) \triangleq \underset{n}{\sup} \  v^n(y) < \infty$ for all $y > 0$.  Then for all $x \in \mb{R}$, $u^*(x) \triangleq \underset{n}{\sup} \  u^n(x) < 0$.
\end{lemma}

\begin{proof} By passing to a subsequence, we can assume that $u^n(0) \ra u^*(0)$.  For each $n$, $u^n(x) = \exp(-x)u^n(0)$, and similarly for $u^*(x)$.  Hence, $u^n \ra u^*$ locally uniformly and $u^*$ is concave.  Let $\overline{v}$ be the convex dual of $u^*$.  Since $v^n$ and $u^n$ are convex duals, then $\lim_n v^n$ exists and is the convex dual of $u^*$, and hence is equal to $\overline{v}$.  By definition, $\overline{v} \leq v^*$.  Suppose that for some $x$, $u^*(x) = 0$.  Then $u^* \equiv 0$.  But, if $u^* \equiv 0$, then it would be that $\overline{v}(y) = \underset{x \in \mb{R}}{\sup} [u^*(x) - xy] \equiv \infty$, which contradicts the finiteness of $v^*(y)$.  Thus, $u^*(x)$ is bounded away from zero.
\end{proof}
Let 
\[x + \widehat{X}^n \triangleq x + \widehat{X}^n(0) = \widehat{X}^n(x)
\]
be the optimal wealth process in market $n$ from initial capital $x$.  This special form for the optimal wealth processes is due to the wealth homogeneity of the exponential utility.  Let $\mc{T}$ be the set of $[0,T]$-valued stopping times.
\begin{proposition}\label{prop3} Suppose that $\underset{n}{\sup} \ ||\lambda^n \cdot M||_{bmo_2} < \infty$.  Then $\{ \exp(-\widehat{X}^n_\tau) : n \in \mb{N}, \tau \in \mc{T} \}$ is uniformly integrable.
\end{proposition}
\begin{proof}
Recall $\widehat{Z}^n$ is the density of the minimal entropy martingale measure for $S^n$, which we denote by $\widehat{\mb{Q}}^n$.  From Theorem $2.2$ of \cite{MR1865021}, $\widehat{X}^n$ is a true $\widehat{\mb{Q}}^n$-martingale for each $n$.  From Theorem $2.2$ of \cite{MR1865021} again, we have $c_n e^{-\widehat{X}^n_T} = \widehat{Z}^n_T$ for some constant $c_n$.

Taking conditional expectations under $\widehat{\mb{Q}}^n$ via Bayes' rule, and using the fact that $\widehat{X}^n$ is a $\widehat{\mb{Q}}^n$-martingale, we obtain
\begin{eqnarray*}
\lefteqn{\log c_n - \widehat{X}^n_\tau} \\
&& = E^{\widehat{\mb{Q}}^n}\left[\log c_n - \widehat{X}^n_T | \mc{F}_\tau \right] \\
&& = E^{\widehat{\mb{Q}^n}} \left[\log \widehat{Z}^n_T | \mc{F}_\tau \right] \\
&& = E \left[\frac{\widehat{Z}^n_T}{\widehat{Z}^n_\tau} \log \widehat{Z}^n_T \ \Bigg| \ \mc{F}_\tau \right] \\
&& = E \left[\frac{\widehat{Z}^n_T}{\widehat{Z}^n_\tau} \left( \log  \frac{\widehat{Z}^n_T}{\widehat{Z}^n_\tau} + \log \widehat{Z}^n_\tau \right) \ \Bigg| \ \mc{F}_\tau \right] \\
&& = E \left[\frac{\widehat{Z}^n_T}{\widehat{Z}^n_\tau} \log \frac{\widehat{Z}^n_T}{\widehat{Z}^n_\tau} \ \Bigg| \  \mc{F}_\tau \right] + \log \widehat{Z}^n_\tau.
\end{eqnarray*}

Exponentiating the previous inequality, we obtain 
\begin{eqnarray*}
\lefteqn{\exp(-\widehat{X}^n_\tau)} \\
&& = \frac{1}{c_n}\widehat{Z}^n_\tau \exp \left( E \left[\frac{\widehat{Z}^n_T}{\widehat{Z}^n_\tau} \log \frac{\widehat{Z}^n_T}{\widehat{Z}^n_\tau} \ | \mc{F}_\tau \right] \right) \\
&& \leq \frac{1}{c_n}e^{\widehat{K}^n_{LLogL} + 1} \widehat{Z}^n_\tau,
\end{eqnarray*}
where $\widehat{K}^n_{LLogL}$ is the $\mc{R}_{LLogL}$ constant of $\widehat{Z}^n$.  According to Proposition \ref{exp_prop1} and Lemma \ref{exp_lemma1}, $\underset{n}{\sup} \ \widehat{K}^n_{LLogL} < \infty$.  By Corollary \ref{exp_cor1}, $v^*(y) < \infty$, and so Lemma \ref{exp_lemma4} implies that $u^* < 0$.  Note that $c_n = -u^n(0)$.  Thus, $\underset{n}{\inf} \ c_n > 0$, so that $\underset{n}{\sup} \ \frac{1}{c_n} < \infty$.  We may then write 
\begin{equation}\label{eq:wsuiorhs}
\exp(-\widehat{X}^n_\tau) \leq \textbf{C} \widehat{Z}^n_\tau
\end{equation} 
\noindent for some constant $\textbf{C}$, so that the inequality is valid for all $n$ and all $\tau$. In what follows we will show that the right-hand-side of  \eqref{eq:wsuiorhs} is uniformly integrable, which completes the proof. Since $\underset{n}{\sup} \ E \left[ V(\widehat{Z}^n_T) \right] < \infty$ (thanks to $V$-compactness and Lemma~\ref{exp_lemma1}) and $V(x)/x \ra \infty$ as $x \ra \infty$, the de la Vall\'{e}e-Poussin criterion implies that $\{\widehat{Z}_T^n : n \in \mb{N} \}$ is uniformly integrable.  Since each $\widehat{Z}^n$ is a martingale, this extends to the uniform integrability of $\{ \widehat{Z}^n_\tau : n \in \mb{N}, \tau \in \mc{T} \}$.
\end{proof}

\begin{remark} In the literature (see \cite{MR2152241}), admissible wealth processes are sometimes defined directly to be those satisfying the conclusion of Proposition \ref{prop3}, i.e. having uniformly integrable utility over all stopping times.
\end{remark}
 
For $i \in \mb{Z}$, let $\widehat{\tau}^{(n,i)} \triangleq \inf \{t : \widehat{X}^n_t = i \}$, and let $\widehat{X}^{(n,i)} \triangleq (\widehat{X}^n)^{\widehat{\tau}^{(n,i)}}=\left(\widehat{X}^n_{{\widehat{\tau}^{(n,i)} \wedge t}}\right)_{t \in [0,T]}$.  

\begin{lemma}\label{exp_lemma5} Suppose that $\underset{n}{\sup} \ ||\lambda^n \cdot M||_{bmo_2} < \infty$.  Then for each $i \in \mb{N}$, the collection $\{ (\widehat{X}^{(n,i)})^* : n \in \mb{N} \}$ is bounded in probability.
\end{lemma}

\begin{remark}\label{lemma5remark}
The conclusions of Lemma \ref{exp_lemma5} and Proposition \ref{prop3} will be shown to be sufficient for obtaining continuity of the utility maximization problems.  Given the strength of the $bmo$ hypothesis, it is natural to ask whether these conditions are also necessary.  In Appendix $A$, it is shown that the conclusion of Lemma \ref{exp_lemma5} is indeed necessary.  The conclusion of Proposition \ref{prop3}, however, is not, and it in fact may fail within a \emph{single} market.  We give an example of this in Appendix $B$.  Note that this market, and indeed all continuous markets, still satisfy the local $bmo$ hypothesis of Assumption \ref{regass2}.
\end{remark}

\begin{proof} Let $\mb{Q}^n$ be the probability measure associated to the minimal martingale $Z^n$, which is continuous.  By Corollary \ref{exp_cor1}, each $\mb{Q}^n$ has finite entropy.  Theorem $1$ of \cite{MR2014244} implies that $\widehat{X}^n$ is a $\mb{Q}^n$-martingale for each $n$.  Then it is also true that $\widehat{X}^{(n,i)}$ is a $\mb{Q}^n$-martingale for each $n$.  Since $\underset{n}{\sup} \  ||\lambda^n \cdot M||_{bmo_2} < \infty$, Lemma \ref{lemmabmoholder} implies that there exists a $p>1$ such that each $Z^n$ satisfies the Reverse H\"{o}lder inequality $\mathcal{R}_p(P)$ with uniformly bounded constant $C_p$.  

By Proposition \ref{exp_prop4}, for $q = \frac{p}{p-1}$,
\begin{eqnarray*}
\lefteqn{\lambda^q P \left((\widehat{X}^{(n,i)})^* > \lambda \right)} \\
&&  \leq C_p E \left[ \left| \widehat{X}^{(n,i)}_T \right|^q \right] \\
&& \leq C_p \left( i^q + \textbf{C}_i E \left[ \exp \left( -\widehat{X}^{(n,i)}_T \right) \right] \right) \\
&& \leq C_p(i^q + \widetilde{\textbf{C}}_i),
\end{eqnarray*}
for constants $\textbf{C}_i, \widetilde{\textbf{C}}_i$ independent of $n$, with the third inequality a consequence of Proposition \ref{prop3}.  
\end{proof}
\begin{proposition}\label{lastprop}  Suppose $\underset{n}{\sup} \ ||\lambda^n \cdot M||_{bmo_2} < \infty$.  Then $u^{(n,k)} \ra u^n$ as $k \ra \infty$, uniformly over the markets $n$.
\end{proposition}

\begin{remark} As indicated by Proposition \ref{unifco} in Appendix $A$, the uniform approximation condition given above is both necessary and sufficient for convergence of the utility maximization problem.
\end{remark}

\begin{proof}  
Let $\epsilon > 0$.  Fix $i \in \mb{N}$ large enough so that $0 > -\exp(-i) > - \epsilon$.  Then 
\begin{equation}\label{preq1}
\begin{split}
 u^n(0) & \geq E \left[ U(\widehat{X}^{(n,i)}_T) \right] \\
& > u^n(0) - \epsilon \text{ for all } n \in \mb{N}. 
\end{split}
\end{equation}
 For $k \in \mb{N}$, let $\widehat{X}^{(n,i,-k)} \triangleq (\widehat{X}^n)^{\widehat{\tau}^{(n,-k)} \wedge \widehat{\tau}^{(n,i)}} = (\widehat{X}^{(n,i)})^{\widehat{\tau}^{(n,-k)}}$.  We claim that 
\begin{equation}\label{preq2}
\underset{k \ra \infty}{\lim} \underset{n \in \mb{N}}{\sup} \ P \left(\widehat{\tau}^{(n,-k)} < \widehat{\tau}^{(n,i)} \right) = 0.
\end{equation}
\noindent Indeed, Lemma \ref{exp_lemma5} implies that the collection $\{ (\widehat{X}^{(n,i)})^* : n \in \mb{N} \}$ is bounded in probability.  Therefore, $\underset{k \ra \infty}{\lim} \underset{n}{\sup} \ P((\widehat{X}^{(n,i)})^* \geq k) = 0$.  But $P(\widehat{\tau}^{(n,-k)} < \widehat{\tau}^{(n,i)}) \leq P((\widehat{X}^{(n,i)})^* \geq k)$, which establishes \eqref{preq2}.
  We next claim that 
\begin{equation}\label{preq2a}
\underset{k \ra \infty}{\lim} \underset{n}{\sup} \ \bigg| E \left[U(\widehat{X}_T^{(n,i)}) \right] - E \left[ U(\widehat{X}^{(n,i,-k)}_T) \right] \bigg| = 0.
\end{equation}
 Let $\epsilon_2 > 0$.  Write 
\begin{eqnarray*}
\lefteqn{E \left[U(\widehat{X}_T^{(n,i,-k)}) \right]} \\
&& = E \left[ U(\widehat{X}_T^{(n,i)}) 1_{\{\widehat{\tau}^{(n,-k)} \geq \widehat{\tau}^{(n,i)}\}} + U(\widehat{X}^{(n,i,-k)}_T) 1_{\{\widehat{\tau}^{(n,-k)} < \widehat{\tau}^{(n,i)} \}} \right] \\
&& = E \left[U(\widehat{X}_T^{(n,i)}) -  U(\widehat{X}_T^{(n,i)}) 1_{\{\widehat{\tau}^{(n,-k)} < \widehat{\tau}^{(n,i)}\}} + U(\widehat{X}^{(n,i,-k)}_T) 1_{\{\widehat{\tau}^{(n,-k)} < \widehat{\tau}^{(n,i)} \}}\right].
\end{eqnarray*}
According to Proposition \ref{prop3}, the set $\{\exp(-\widehat{X}^n_\tau) : n \in \mb{N}, \tau \in \mc{T} \}$ is uniformly integrable, which immediately implies that the set $\{\exp(-\widehat{X}^{(n,i)}_T), \exp(-\widehat{X}^{(n,i,-k)}_T) : n,k \in \mb{N}\}$ is uniformly integrable.  Therefore, there exists $\delta = \delta(\epsilon_2) > 0$ such that for any set $A$, $P(A) < \delta$ implies that $\max \left \{ E[U(\widehat{X}_T^{(n,i)})1_A], E[U(\widehat{X}^{(n,i,-k)}_T)1_A] \right \}< \epsilon_2$.  According to $\eqref{preq2}$, there exists $k_0 \in \mb{N}$ such that for $k \geq k_0$ and all $n \in \mb{N}$, the sets $\{\widehat{\tau}^{(n,-k)} < \widehat{\tau}^{(n,i)}\}$ have probability less than $\delta$.  Therefore, for $k \geq k_0$ and all $n \in \mb{N}$, $\max\left\{E \left[U(\widehat{X}_T^{(n,i)}) 1_{\{\widehat{\tau}^{(n,-k)} < \widehat{\tau}^{(n,i)}\}} \right], E \left[U(\widehat{X}^{(n,i,-k)}_T)1_{\{\widehat{\tau}^{(n,-k)} < \widehat{\tau}^{(n,i)}\}} \right] \right\}< \epsilon_2$.  Thus, for $k \geq k_0$ and all $n \in \mb{N}$, we have 
\[
\bigg| E \left[U(\widehat{X}_T^{(n,i)}) \right] - E \left[ U(\widehat{X}^{(n,i,-k)}_T) \right] \bigg| < 2\epsilon_2,
\]
and \eqref{preq2a} is established.  Then \eqref{preq1} and \eqref{preq2a} imply that 
\begin{equation}\label{preq3}
\underset{k \ra \infty}{\lim} \ \underset{n \in \mb{N}}{\sup} \ \bigg| \ u^n(0) - E \left[ U(\widehat{X}^{(n,i,-k)}_T) \right] \bigg| \leq \epsilon.
\end{equation}
  
\noindent Since $\widehat{X}^{(n,i, -k)} > -k-1$, by definition, $u^{(n,k)}(0) \geq \left[ U(\widehat{X}_T^{(n,i,-k)}) \right]$.  Then \eqref{preq3} and the fact that $u^{(n,k)} \leq u^n$ imply that for any $\epsilon > 0$,
\begin{equation}\label{preq4}
\underset{k \ra \infty}{\lim} \ \underset{n \in \mb{N}}{\sup} \ |u^n(0) - u^{(n,k)}(0)| \leq \epsilon,
\end{equation}

\noindent implying that $\underset{k \ra \infty}{\lim} \ \underset{n \in \mb{N}}{\sup} \ |u^{(n,k)}(0) - u^n(0)| = 0$, i.e. that $u^{(n,k)} \ra u^n$ as $k \ra \infty$, uniformly over $n$.
\end{proof}

\subsection{Proof of the intermediate theorem}

\begin{proof}[Proof of Theorem \ref{mainthm}]
It follows from Corollary~\ref{exp_cor1}  and Lemma~\ref{exp_lemma3} that for each $k$, $\underset{n \ra \infty}{\lim} u^{(n,k)} = u^{(\infty,k)}$.  Proposition \ref{lastprop}, on the other hand, states that $\underset{k \ra \infty}{\lim} u^{(n,k)} = u^n$, uniformly over $n$.  These facts together imply that $\underset{n \ra \infty}{\lim} u^n = u^\infty$.
\end{proof}

%\begin{remark}\label{rem:uniformapprox}
%Since $u^{n,k}$ is increasing in $k$, a double sequence version of Dini's theorem implies that $u^{n} \ra u^{\infty}$  if and only if $u^{n,k} \to u^{n}$, uniformly over $n$. 
%That is, for the stability result in Theorem~\ref{mainthm} it is necessary and sufficient that one can approximate the value functions $u^n$ uniformly by bounding the wealth processes from below. The bmo hypothesis is the key to establishing this uniform approximation property, as can be seen from Proposition~\ref{lastprop}.
%\end{remark}

\section{Proofs of the Main Theorems}\label{exp_prf}

We establish the main Theorem \ref{fin.mainthm} in pieces, establishing lower semi-continuity and upper semi-continuity separately.  The proof of lower semi-continuity is the easier of the two, and indeed is not dependent on the special structure of the exponential utility function.

\begin{lemma}\label{lsclemma}  Suppose that $Z^n_T \ra Z^\infty_T$ in probability and that $\left \{Z_T^n : n \in \mb{N} \right \}$ is $V$-compact, i.e. Assumption \ref{regass1} holds.  Then $u^\infty(x) \leq \underset{n \ra \infty}{\liminf} \  u^n(x)$.
\end{lemma}

\begin{proof} As in the proof of Lemma \ref{exp_lemma3}, the $V$-compactness of $ \left \{Z^n_T : n \in \mb{N} \right \}$ implies that this set is also $V^{(k)}$-compact, where $V^{(k)}$ is the dual of the ``truncated" utility function $U^{(k)} \leq U$, defined at the beginning of Section \ref{exp_approx}.  By Lemma \ref{exp_lemma3} and the main theorem of \cite{MR2438002}, $\underset{n \ra \infty}{\lim} \  u^{(n,k)}(x) = u^{(\infty,k)}(x)$ for each $k \in \mb{N}$. By Step $1$ of Theorem $2.2$ of \cite{MR1865021}, $u^n(x) = \underset{k \in \mb{N}}{\sup} \ u^{(n,k)}(x)$ for each $n$.  Therefore 
\[
\begin{split}
\underset{n \ra \infty}{\liminf} \ u^n(x) & = \underset{n \ra \infty}{\liminf} \  \underset{k \in \mb{N}}{\sup} \ u^{(n,k)}(x) \\
& \geq \underset{k \in \mb{N}}{\sup} \ \underset{n \ra \infty}{\liminf} \  u^{(n,k)}(x) \\
& = \underset{k \in \mb{N}}{\sup} \ u^{(\infty,k)}(x) \\
& = u^\infty(x).
\end{split}
\]  
\end{proof}

We now establish upper semi-continuity.

\begin{proposition}\label{usclemma}  Suppose that there exists a sequence of stopping times $(\tau_j) \uparrow T$ such that for each $j$, $\underset{n}{\sup} \ ||(\lambda^n \cdot M)^{\tau_j}||_{bmo} < \infty$, i.e. Assumption \ref{regass2} holds.  Additionally, suppose that $V(Z^\infty_T) \in L^1$ and $Z^\infty$ is a martingale. Then $u^\infty(x) \geq \underset{n \ra \infty}{\liminf} \  u^n(x)$.
\end{proposition}

\begin{proof}

For $j,n=1,\ldots,\infty$, let $u^{n,j}$ denote the indirect utility arising from trading in market $n$ up until time $\tau_j$, where $u^n = u^{n,\infty}$.  Since all trading opportunities arising on $[0,\tau_j)$ are also available over the whole time period $[0,T]$, we know that $u^{n,j} \leq u^{n,j+1} \leq u^{n,\infty}$.  We claim that in addition, 
\begin{equation}\label{lasteq2}
u^{\infty,j} \uparrow u^\infty
\end{equation}
 as $j \ra \infty$.  As $Z^\infty$ is a martingale and $V$ is convex, $V(Z^\infty)$ is a submartingale (whose terminal value is integrable).  As $V$ is bounded from below, this implies that $V(Z^\infty)$ is of Class D, as defined, in \cite{MR745449}, p.11, for example.  In particular, the set $\left \{ V \left(Z^\infty_{\tau_j} \right ) : j \in \mb{N} \right \}$ is uniformly integrable.  In the context of Lemma \ref{lsclemma}, set $Z^j \triangleq (Z^\infty)^{\tau_j}$, so that $Z^\infty_{\tau_j} = Z^j_T$.  So, applying Lemma \ref{lsclemma} to the sequence $\left \{ Z^j \right \}$, it follows that $u^\infty \leq \underset{j \ra \infty}{\liminf} \  u^{\infty,j}$.  As $u^{\infty,j} \leq u^{\infty,j+1} \leq u^\infty$, \eqref{lasteq2} now follows.  

We now claim that for each $j < \infty$, $u^{n,j} \ra u^{\infty,j}$.  First, $Z^n_T \ra Z^n_\infty$ in $L^1$ by Scheffe's Lemma, and hence $Z^n \ra Z^\infty$ in ucp, which follows by applying Doob's weak $L^1$ inequality.  In particular, $Z^n_{\tau_j} \ra Z^\infty_{\tau_j}$ in probability.  We are now in the setting of Theorem \ref{mainthm}: considering $\tau_j$ as our terminal time, we have $Z^n_{\tau_j} \ra Z^\infty_{\tau_j}$ in probability.  Note that Theorem \ref{mainthm} can be applied to the terminal time $\tau_j \leq T$ by considering, for example, $(Z^n)^{\tau_j}$ defined on the time interval $[0,T]$.  So, applying Theorem \ref{mainthm}, we deduce that for each $j < \infty$, $u^{n,j} \ra u^{\infty,j}$.  

Next, we claim that 
\begin{equation}\label{lasteq3}
\liminf_{n \ra \infty} u^n \geq u^\infty.
\end{equation}
First choose $\epsilon >0$ and $J=J(\epsilon)$ sufficiently large so that  $u^{\infty,J} \geq u^\infty - \epsilon$.  Next choose $N = N(J)$ such that, for $n \geq N$, $|u^{n,J} - u^{\infty,J}| < \epsilon$.  The triangle inequality implies that $|u^{n,j} - u^\infty| < 2 \epsilon$.  Since $u^n \geq u^{n,j}$, it follows that for $n \geq N$, $u^n \geq u^\infty - 2\epsilon$.  In other words, $\liminf_{n \ra \infty} u^n \geq u^\infty$, and \eqref{lasteq3} is established.
\end{proof}

We now obtain:

\begin{proof}[Proof of Theorem \ref{fin.mainthm}]
By Lemma \ref{lsclemma}, the mapping is lower semi-continuous, and by Proposition \ref{usclemma}, the mapping is upper semi-continuous.  Together, these imply the theorem.
\end{proof}

\begin{proof}[Proof of Theorem \ref{fin.mainthm2}]

As before, for each $k \in \mb{N}$, let $\widehat{X}^{n,k}_T$ be the optimal terminal wealth in the $nth$ market that satisfies the constraint $\widehat{X}^{n,k}_T > -k$.  By Step $7$ in the proof of Theorem $2.2$ of \cite{MR1865021}, we know that as $k \ra \infty$, $U(\widehat{X}^{n,k}_T) \ra U(\widehat{X}^n_T)$ in $L^1$ for each $n \in \mb{N}$.  As a consequence of Proposition \ref{lastprop}, $E[U(\widehat{X}^{n,k}_T)] \uparrow E[U(\widehat{X}^n_T)]$ as $k \ra \infty$, and the convergence is uniform over $n$.  As $U(\cdot)$ is nonpositive, Scheffe's Lemma then implies that $U(\widehat{X}^{n,k}_T) \ra U(\widehat{X}^n_T)$ in $L^1$ as $k \ra \infty$, \emph{uniformly over $n$}.  $L^1$ convergence being stronger than $L^0$ convergence, we also have that $U(\widehat{X}^{n,k}_T) \ra U(\widehat{X}^n_T)$ in probability as $k \ra \infty$, uniformly over $n$.  Since $\widehat{X}^{n,k}_T \ra \widehat{X}^{\infty,k}_T$ in probability as $n \ra \infty$ for all $k$, then by Lemma $3.10$ of \cite{MR2438002}, it follows that $U(\widehat{X}^n_T) \ra U(\widehat{X}^\infty_T)$ in probability.  Since $U$ is bounded from above, we need a little more work to show that $\widehat{X}^n_T \ra \widehat{X}^\infty_T$ in probability.

We claim now that $\{\widehat{X}^n_T \}_{n \in \mb{N}}$ is bounded in probability.  Note that $Z^n \widehat{X}^n$ is a martingale, so $E[Z^n_T \widehat{X}^n_T] = 0$.  By Proposition \ref{prop3}, $\{U(\widehat{X}^n_T)\}_{n \in \mb{N}}$ is uniformly integrable, and $V$-compactness implies that $\{V(Z^n_T)\}_{n \in \mb{N}}$ is uniformly integrable.  The duality relationship $Z^n_T \widehat{X}^n_T \geq U(\widehat{X}^n_T) - V(Z^n_T)$ now implies that the negative parts $\{(Z^n_T \widehat{X}^n_T)^- \}_{n \in \mb{N}}$ are uniformly integrable.  Hence $\{Z^n_T \widehat{X}^n_T \}_{n \in \mb{N}}$ is bounded in $L^1$, and also in $L^0$.  But $Z^n_T \ra Z^\infty_T$ in probability, and $Z^\infty_T$ is strictly positive.  Hence $\{Z^n_T\}_{n \in \mb{N}}$ is bounded away from zero in probability, and it follows that $\{\widehat{X}^n_T \}_{n \in \mb{N}}$ is bounded in probability.

Suppose now that $\widehat{X}^n_T$ does not converge to $\widehat{X}^\infty_T$ in probability.  Then there exists an $\epsilon > 0$ such that for infinitely many $n$,
$P(|\widehat{X}^n_T - \widehat{X}^\infty_T| > \epsilon) > \epsilon$.  Now, choose a compact set $K$ such that $P(\widehat{X}^n_T \not \in K) < \frac{\epsilon}{4}$ for all $n$.  Then \\ $P \left( |\widehat{X}^n_T - \widehat{X}^\infty_T|>\epsilon, \text{ and } \widehat{X}^n_T,\widehat{X}^\infty_T \in K \right) > \frac{\epsilon}{2}$.  For $x,y \in K$, there exists a constant $c>0$ such that $|U(x) - U(y)| > c|x-y|$, due to the fact that $U'(x)$ is positive and bounded away from zero on the compact set $K$.  Thus, it follows that for infinitely many $n$, $P(|U(\widehat{X}^n_T) - U(\widehat{X}^\infty_T)| > c\epsilon) > \frac{\epsilon}{2}$, contradicting the fact that $U(\widehat{X}^n_T) \ra U(\widehat{X}^\infty_T)$ in probability.
\end{proof}

\section{On Assumptions \ref{regass1} and \ref{regass2}}\label{exp_ass}

\subsection{Comparison of \ref{regass2} with the Half-Line setting}

Recall that in addition to the $V$-compactness assumption \ref{regass1}, we required Assumption \ref{regass2}, which has no direct analog in \cite{MR2438002}.  Reviewing Proposition \ref{prop3}, one sees that the purpose of this assumption was to ensure that (locally) the set $\left \{\exp \left(-\widehat{X}^n_\tau \right) : n \in \mb{N}, \tau \in \mc{T} \right \}$ is uniformly integrable.  More precisely, when we say locally, we mean that there exists a sequence of stopping times $\tau_j \uparrow T$ such that $\left \{\exp \left(-\widehat{X}^n_{\tau \wedge \tau_j} \right) : n \in \mb{N}, \tau \in \mc{T} \right \}$ is uniformly integrable for each $j$.

Indeed, Assumption \ref{regass2} could be weakened so that it is exactly this condition: let $C_j \triangleq \underset{\tau_j \geq \tau \in \mc{T}, n \in \mb{N}}{\es} \ E \left[ \frac{Z^n_{\tau_j}}{Z^n_\tau} \log \frac{Z^n_{\tau_j}}{Z^n_\tau} | \mc{F}_\tau \right]$, which is uniformly bounded thanks to Assumption \ref{regass2} and Proposition \ref{exp_prop1}.  We have, as in the proof of Proposition \ref{prop3}, for $\tau \leq \tau_j$, the duality relationship

\begin{equation}\label{dualityeq}
\exp(-\widehat{X}^n_\tau) = \underbrace{\frac{1}{c_n}\widehat{Z}^n_\tau \exp \left( E \left[\frac{\widehat{Z}^n_{\tau_j}}{\widehat{Z}^n_{\tau}} \log \frac{\widehat{Z}^n_{\tau_j}}{\widehat{Z}^n_\tau} \ | \mc{F}_\tau \right] \right).}_{ \leq \frac{1}{c_n} \widehat{Z}^n_\tau \exp \left(C_j \right) \text{ by } Assumption \  \ref{regass2} }
\end{equation}

Now, the actual structural condition we need to prove our main results is that \\ $\left \{\exp \left(-\widehat{X}^n_{\tau \wedge \tau_j} \right) : n \in \mb{N}, \tau \in \mc{T} \right \}$ is uniformly integrable.  By $V$-compactness, the set $\{\widehat{Z}^n_\tau : \tau \in \mc{T}, n \in \mb{N} \}$ is uniformly integrable, as in Proposition \ref{prop3}.  Therefore, from considering \eqref{dualityeq}, we see that Assumption \ref{regass2} implies the uniform integrability of $\left \{\exp \left(-\widehat{X}^n_{\tau \wedge \tau_j} \right) : n \in \mb{N}, \tau \in \mc{T} \right \}$; additionally, we see that Assumption \ref{regass2} is a slightly stronger hypothesis than the required property of uniform integrability of $\left \{\exp \left(-\widehat{X}^n_{\tau \wedge \tau_j} \right) : n \in \mb{N}, \tau \in \mc{T} \right \}$ for each $j$.

This uniform integrability condition is useful because it allows for processes $\widehat{X}^{n,k}$ such that $\underset{ k \ra \infty}{\lim} E \left[ -\exp \left( - \widehat{X}^{n,k}_T \right) \right] = E \left[ -\exp \left(-\widehat{X}^n_T \right) \right]$, uniformly over $n$, and \\ $\underset{ 0 \leq t \leq T, n \in \mb{N}}{\sup} \exp \left( -\widehat{X}^{n,k}_t \right) \in L^\infty$.  Here, we show that, in the setting of utility maximization with a utility function $\widetilde{U}$ defined on $\mb{R}_+$, this uniform approximation property is already implied by the $\widetilde{V}$-compactness assumption, with $\widetilde{V}$ the conjugate of $\widetilde{U}$.  The proof of this fact is interesting because it mirrors, in our opinion, the essential technical step of \cite{MR2438002}, Corollary $3.4$.

\begin{proposition} In the setting of \cite{MR2438002}, let $\left \{Z^n : n = 1,2,\ldots,\infty \right \}$ define a $\widetilde{V}$-compact sequence of markets with $Z^n_T \ra Z^\infty_T$ in probability.  Fix an initial wealth $x_0$ , and let $\widetilde{X}^n$ be the optimal wealth process starting from $x_0$ in the $n^{th}$ market.  Then there exist wealth processes $\widetilde{X}^{n,k}$, each defined in the $n^{th}$ market, such that $E \left[ \widetilde{U} \left(\widetilde{X}^{n,k}_T \right) \right] \ra E \left[ \widetilde{U} \left(\widetilde{X}^n_T \right) \right]$ as $k \ra \infty$, uniformly over all $n$, and \\ $\underset{0 \leq t \leq T, n \in \mb{N}}{\sup} \widetilde{U}^-\left(\widetilde{X}^{n,k}_t \right) \in L^\infty$.
\end{proposition}

\begin{proof}  The proposition uses a simple construction, inspired by \cite{MR2438002}.  Given $\widetilde{X}^n$, define 
\[
\widetilde{X}^{n,k} \triangleq \frac{1}{k}x_0 + \frac{k-1}{k} \widetilde{X}^n.
\]
In other words, the wealth process $\widetilde{X}^{n,k}$ follows the optimal trajectory, except that a small portion is set aside and left in the riskless asset.  The concavity of $\widetilde{U}$ implies that 
\begin{eqnarray}\label{lasteq4}
\lefteqn{\frac{1}{k}\widetilde{U}(x_0) + \frac{k-1}{k} E \left[ \widetilde{U} \left(\widetilde{X}^n_T \right) \right]} \\
&& \leq E \left[ \widetilde{U}\left(\widetilde{X}^{n,k}_T \right) \right] \nonumber \\
&& \leq E \left[ \widetilde{U} \left(\widetilde{X}^n_T \right) \right]. \nonumber
\end{eqnarray}

  Since $\widetilde{V}$-compactness implies that the collection $\left \{\widetilde{U}\left(\widetilde{X}^n_T \right) : n \in \mb{N} \right \}$ is bounded in $L^1$, the uniform approximation property is established in \eqref{lasteq4}.  Next, each wealth process $\widetilde{X}^n$ is strictly positive, and therefore $\widetilde{X}^{n,k} > \frac{1}{k}$.  Consequently
\[
\underset{0 \leq t \leq T, n \in \mb{N}}{\sup} \widetilde{U}^-\left(\widetilde{X}^{n,k}_t \right) < \widetilde{U}^-\left(\frac{1}{k} \right).
\]
\end{proof}

\subsection{Economic interpretation of Assumption \ref{regass2}}

Consider a generic market with dynamics $S= M + \int \lambda d \langle M \rangle$ and associated minimal martingale measure $Z= \mc{E}(-\lambda \cdot M)$.  In this market, consider the opportunity process $L_t^{exp}$, introduced in \cite{MR2721695}, and used in \cite{MR2892960}.  It is the utility value process normalized by the optimal wealth process, and it exists as a consequence of the homogeneity of power and exponential utilities, and their associated optimal wealth processes.  In the notation of \cite{MR2721695}, the opportunity process $L_t^{exp}$ satisfies
\[
V_t(\theta) = \exp(-G_t(\theta))L_t^{exp},
\]

\noindent where $V_t(\theta)$ represents the indirect utility arising from following the trading strategy $\theta$ up to time $t$, and $G_t(\theta)$ is the wealth resulting from trading according to $\theta$ up to time $t$.  As the name suggests, $L_t^{exp}$ describes how much utility can be attained per unit of wealth.  Then equation $(6.6)$ of \cite{MR2721695} establishes a relationship between $L^{exp}$ and $\widehat{Z}$:
\begin{equation}\label{oppdef}
-\log (L_t^{exp}) = E \left[ V \left( \frac{\widehat{Z}_T}{\widehat{Z}_t} \right) | \mc{F}_t \right].
\end{equation}

Frequently in this paper, we have concerned ourselves with the size of the right hand side of \eqref{oppdef}: specifically, the $bmo$ hypothesis has been used to establish a uniform upper bound on this term over $t$.  This implies that the value processes $L_t^{exp}$ are uniformly bounded away from zero.  In economic terms, this puts a constraint on how attractive the investment opportunities can be in our sequence of markets.  If the opportunity process is close to zero, this means that an optimal investing agent is relatively unconcerned with having very negative wealth, in that $L_t^{exp}$ counteracts the size of $\exp(-G_t(\theta))$ (note that we wish to maximize $V_t(\theta)$, which is negative).

Note, however, that in \eqref{oppdef}, what matters is the $\mc{R}_{L Log L}$ constant $K_{L Log L}(\widehat{Z})$ for the optimal dual variable $\widehat{Z}$, while our regularity Assumption \ref{regass2} involves the $\mc{R}_{L Log L}$ constant $K_{L Log L}(Z)$ for $Z$, the minimal martingale measure.  By Lemma \ref{exp_lemma1}, we know that $K_{L Log L}(\widehat{Z}) \leq K_{L Log L}(Z)$.  More interesting is the claim that the sizes of $K_{L Log L}(\widehat{Z})$ in fact control the sizes of $K_{L Log L}(Z)$, a result which we will establish in Proposition \ref{kprop} below.  Thus, the $\mc{R}_{L Log L}$ constants bind the sizes of the minimal entropy martingale and minimal martingale in a substantive way.  In general, the dual object we are interested in is the minimal entropy martingale, while the dual object which we can describe most explicitly is the minimal martingale.  The claim above, however, implies that the ostensibly more restrictive act of placing a regularity assumption on the minimal martingales is essentially equivalent to placing one on the minimal entropy martingales, both implying control over the size of the opportunity process.

\begin{proposition}\label{kprop} Let $S^n$, $n \geq 1$, describe a sequence of markets, with minimal martingales $Z^n$ and minimal entropy martingales $\widehat{Z}^n$.  Then
\[
\underset{n}{\sup} \ K_{L Log L}(Z^n) < \infty \text{ if and only if } \underset{n}{\sup} \ K_{L Log L}(\widehat{Z}^n) < \infty.
\]
\end{proposition}

\begin{proof}
 The ``$\Rightarrow$" direction is trivial, given Lemma \ref{exp_lemma1}.  We therefore address the ``$\Leftarrow$" condition.  Write $\widehat{Z}^n = \mc{E}(\widehat{R}^n) = \mc{E}(-\lambda^n \cdot M + \widehat{L}^n)$, where $\widehat{L}^n$ is a local martingale orthogonal to $M$.  Thus, $\langle \widehat{R}^n \rangle = \langle -\lambda \cdot M \rangle + \langle \widehat{L}^n \rangle$.  As a consequence,
\[
||-\lambda \cdot M||_{bmo_2} \leq ||\widehat{R}^n||_{bmo_2}.
\]

According to the proof of Lemma \ref{lemmabmoholder}, found in Propositions $5$ and $6$ of \cite{MR544804}, there exists an increasing function $f:\mb{R}_+ \ra \mb{R}_+$ such that for a continuous martingale $M$, $||M||_{bmo_2} \leq x$ implies that $K_{L Log L}(\mc{E}(M)) \leq f(x)$.  Therefore, for each $n$, $K_{L Log L}(Z^n) \leq f(||\lambda^n \cdot M||_{bmo_2}) \leq f(||\widehat{R}^n||_{bmo_2})$.  Taking suprema over $n$, we have 
\[
\underset{n}{sup} \  K_{L Log L}(Z^n) \leq \underset{n}{sup} \  f(||\widehat{R}^n||_{bmo_2}) \triangleq R^* < \infty,
\]
with the finite constant $R^*$ existing by hypothesis.

\end{proof}

\subsection{On Assumption \ref{regass1}}

Here, we illustrate the necessity of the $V$-compactness hypothesis with a few examples.

\begin{lemma}\label{lemmacoincide}  Suppose that $Z^n_T \ra Z^\infty_T$ in probability and that $\left \{Z_T^n : n \in \mb{N} \cup \{\infty \} \right \}$ is $V$-compact.  Additionally, suppose that $Z^\infty_T = \widehat{Z}^\infty_T$, i.e. the terminal values of the minimal martingale measure and minimal entropy martingale measure coincide.  Then $\underset{n \ra \infty}{\lim} \ v^n(y) = v^\infty(y)$.  Hence, $\underset{n \ra \infty}{\lim} \ u^n(x) = u^\infty(x)$.
\end{lemma}

\begin{proof} Thanks to Lemma \ref{lsclemma}, it suffices to show that $\underset{n \ra \infty}{\limsup} \  v^n(y) \leq v^\infty(y)$.  By hypothesis, $E[V(yZ^n_T)] \ra E[V(yZ^\infty_T)] = v^\infty(y)$ as $n \ra \infty$.  But $v^n(y) \leq E[V(yZ^n_T)]$.  Therefore, $\underset{n \ra \infty}{\limsup} \ v^n(y) \leq \underset{n \ra \infty}{\lim} \ E[V(yZ^n_T)] = v^\infty(y)$.  The last claim in the lemma, that $u^n \ra u^\infty$, follows from the duality between $v^n$ and $u^n$, see Proposition $3.9$ of \cite{MR2438002}.
\end{proof}

\begin{corollary}\label{completecor}  Suppose that $Z^n_T \ra Z^\infty_T$ in probability and that the limiting market is complete.  Then $ \underset{n \ra \infty}{\lim} \ u^n(x) = u^\infty(x)$ if and only if $\{Z^n_T : n \in \mb{N} \}$ is $V$-compact.
\end{corollary}

\begin{proof}  For the ``if" direction, note that in a complete market there is only one equivalent martingale measure, and hence trivially the minimal martingale measure and minimal entropy martingale must agree.  Therefore, by Lemma \ref{lemmacoincide}, $\underset{n \ra \infty}{\lim} \ u^n(x) = u^\infty(x)$.  The ``only if" direction is identical to the proof of Proposition $2.13$ of \cite{MR2438002}.  
\end{proof}

\begin{remark} We also note that there are examples of incomplete markets where the minimal martingale and minimal entropy martingale agree; in these cases it is also clear that $V$-compactness is necessary and sufficient.  This is the case in a market when one tries to hedge an option written on a non-tradeable asset using a geometric Brownian motion correlated with that asset; see e.g. Section 4 of \cite{jonsir}.
\end{remark}

%\section{Conclusion}
%We have demonstrated that $u^n(x) \ra u^\infty(x)$ when $Z_T^n \ra Z_T^\infty$ in probability together with a slight strengthening of the $V$-compactness %condition.  These conditions are weaker than any in the literature used to show continuity.  In \cite{MR2438002} in the setting of utility functions defined %on $(0,\infty)$, it is shown that the $V$-compactness condition is sufficient for attaining continuity, where $V$ is the convex dual of the general utility %function $U$.  An obvious question is whether $V$-compactness is also sufficient for $U(x) = -\exp(-x)$ or general utility functions on $\mb{R}$.  The proof %from \cite{MR2438002} does not generalize, at least in an obvious way, to this setting.  Crucial to our analysis was the fact that the $bmo$-type hypothesis %placed structure on the whole time interval $[0,T]$, and not just on the time $T$-value.  The utility maximization problem only explicitly takes place at %time $T$, but for utility functions defined on $\mb{R}$, the subtler definition of admissible strategies pulls intermediate values of wealth processes into %the equation.  

\appendix

\section{Continuity and Uniform Approximation}\label{exp_app1}
In this appendix, we address the first claim made in Remark \ref{lemma5remark}.  Its proof requires a bit of preparatory work.

\begin{proposition}\label{unifco}  $u^n \ra u^\infty$ if and only if $u^{(n,k)} \ra u^n$ as $k \ra \infty$, uniformly over $n$.
\end{proposition}

\begin{proof}  The $``\Leftarrow"$ implication was the content of Theorem \ref{mainthm}.  For the other direction, let $\mb{N}^*$ be the space $\{1,2,\ldots,\infty\}$, whose topology is the one point compactification of $\mb{N}$ with the discrete topology; the open sets of $\mb{N}^*$ are the finite subsets of $\mb{N}$ and cofinite subsets containing $\infty$.  This space is compact.  

For each $k \in \mb{N}$, the map $n \mapsto u^{(n,k)}(0)$ is continuous by Lemma \ref{exp_lemma3}.  By construction, $u^{(n,k)} \leq u^{(n,k+1)}$ for all $n,k$, and $u^{(n,k)} \ra u^n$ as $k \ra \infty$ for all $n$.  Therefore, supposing that $n \mapsto u^n(0)$ is continuous, we apply Dini's Theorem to get the desired result.
\end{proof}

\begin{lemma}\label{lemmaa.1} Suppose that $u^n \ra u^\infty$.  Then $\widehat{X}^n_T \ra \widehat{X}^\infty_T$ in probability.  Furthermore, $U(\widehat{X}^n_T) \ra U(\widehat{X}^\infty_T)$ in $L^1$.
\end{lemma}

\begin{proof} The proof of the first claim is identical to Lemma $3.10$ of \cite{MR2438002}, which establishes the result in the positive wealth case.  The second claim follows from Scheffe's Lemma.
\end{proof}
Let $d(\cdot,\cdot)$ be a metric whose topology is associated with the one corresponding to convergence in probability, i.e.
\[
d(X_n, X) \ra 0 \text{ if and only if } P(|X_n - X| > \epsilon) \ra 0 \text{ for all } \epsilon > 0.
\]

Recall that $\widehat{X}^{(n,k)}$ is the optimal wealth process in market $n$ satisfying the constraint $\widehat{X}^{(n,k)} > -k$.

\begin{lemma}\label{lemmaa.2} Suppose that $u^n \ra u^\infty$.  Then 
\[
\underset{k \ra \infty}{\lim} \ \underset{n}{\sup} \ d \left(\widehat{X}^{(n,k)}_T,\widehat{X}^n_T \right)=0.
\]
\end{lemma}

\begin{proof}
On p. 708, Step $2$ of \cite{MR1865021}, it is established that, for fixed $n$, $d \left(\widetilde{Y}^{(n,k)}_T,\widehat{Y}^n_T \right) \ra 0$ as $k \ra \infty$; recall that $\widehat{Y}^n$ is the minimal dual variable arising from utility maximization with $U:\mb{R} \ra \mb{R}$ in the $n^{th}$ market, and $\widetilde{Y}^{(n,k)}$ is the minimal dual variable arising from utility maximization in the $n^{th}$ market with $\widetilde{U}^{(k)}:\mb{R}_+ \ra \mb{R}$, as defined in Section $4$.  For details, we refer the reader to \cite{MR1865021}.  A careful reading of this proof yields the fact that the rate of this convergence, for each $n$, is governed by the rate at which $v^{(n,k)}$ converges to $v^n$.  The hypothesis that $u^n \ra u$ is equivalent, by Proposition \ref{unifco}, to the uniform convergence of $u^{(n,k)}$ to $u^n$ as $k \ra \infty$, over all $n$.  By a standard duality argument, this is equivalent to $v^{(n,k)}$ converging to $v^n$ as $k \ra \infty$, uniformly over all $n$.  
Applying a standard argument based on optimality and strict convexity (see Lemma $3.6$ of \cite{MR1722287}), it therefore follows that
\[
\underset{k \ra \infty}{\lim} \ \underset{n}{\sup} \ d \left(\widetilde{Y}^{(n,k)}_T,\widehat{Y}^n_T \right)=0.
\]

By duality, we have $U^{(k)'}(\widehat{X}^{(n,k)}_T) = u^{(n,k)'}(0)\widetilde{Y}^{(n,k)}_T$, and the lemma follows.
\end{proof}

\begin{corollary}\label{cora.3} Suppose that $u^n \ra u^\infty$.  Then 
\[
\underset{k \ra \infty}{\lim} \ \underset{n}{\sup} \ ||U(\widehat{X}^{(n,k)}_T)- U(\widehat{X}^n_T)||_{L^1}=0.
\]
\end{corollary}

\begin{proof}
The result follows by applying Lemma \ref{lemmaa.2} and Proposition \ref{unifco}, along with Scheffe's Lemma.
\end{proof}

\begin{corollary}\label{cora.4} Suppose $u^n \ra u^\infty$.  Then the set $\{U(\widehat{X}^{(n,k)}_T) : n,k \}$ is uniformly integrable.
\end{corollary}

\begin{proof} The result follows by applying Lemma \ref{lemmaa.1} and Corollary \ref{cora.3}.
\end{proof}

\begin{lemma}\label{lemmaa.5} Suppose that $u^n \ra u^\infty$.  Then
\begin{equation}\label{lemmaa.5eq1}
\underset{k \ra \infty}{\lim} \ \underset{n}{\sup} \ d \left((\widehat{X}^{(n,k)} - \widehat{X}^n)^*,0 \right) = 0.
\end{equation}
\end{lemma}

\begin{proof} Suppose that \eqref{lemmaa.5eq1} does not hold.  Then, there exists a sequence $(n_m,k_m)_{m \geq 1}$ and $\alpha > 0$ such that
\[
P \left( (\widehat{X}^{(n_m,k_m)} - \widehat{X}^{n_m})^* > \alpha \right) > \alpha
\]
for each $m$.  Let $\tau_m = \inf \{ t \geq 0 : \widehat{X}^{(n_m,k_m)}_t \geq \widehat{X}^{n_m}_t + \alpha \} \wedge T$, and let $\widetilde{\tau}_m = \inf \{ t \geq 0 : \widehat{X}^{(n_m,k_m)}_t \leq \widehat{X}^{n_m}_t - \alpha \} \wedge T$.  It must be the case that either $P(\tau_m < T) > \frac{\alpha}{2}$ or $P(\widetilde{\tau}_m < T) > \frac{\alpha}{2}$.  The treatment of each contingency is similar, and so without loss of generality, we assume that $P(\tau_m < T) > \frac{\alpha}{2}$.   Consider the concatenated wealth process $\widetilde{X}^{n_m}_t \triangleq \widehat{X}^{(n_m,k_m)}_{t \wedge \tau_m} + (\widehat{X}^{n_m}_{t \vee \tau_m} - \widehat{X}^{n_m}_{\tau_m})$.  For any $\mb{Q} \in \mc{M}^n$ with finite entropy, $\widehat{X}^{(n_m, k_m)}$ and $\widehat{X}^{n_m}$ are $\mb{Q}$-martingales, since they are admissible wealth processes in the sense of Definition \ref{admdef}.  Since the concatenation of martingales yields a martingale, $\widetilde{X}^{n_m}$ is a $\mb{Q}$-martingale for any $\mb{Q} \in \mc{M}^n$ with finite entropy, so this concatenated strategy is still admissible.  On the set $\{\tau_m < T\}$, $\widetilde{X}^{n_m}_T \geq \alpha + \widehat{X}^{n_m}_T$, and on the set $\{\tau_m < T\}^c$, $\widetilde{X}^{n_m}_T = \widehat{X}^{(n_m, k_m)}_T$.  As in the proof of Theorem \ref{fin.mainthm2}, for any $\epsilon > 0$, there is a compact subset $K = K(\epsilon)$ of $\mb{R}$ such that 
\begin{equation}\label{lasteq5}
\max \left \{ P(\widehat{X}^{n_m}_T - \alpha \not \in K), P(\widehat{X}^{(n_m, k_m)}_T - \alpha \not \in K) \right \}< \epsilon
\end{equation}
for all $m$, and $U'(x) \geq c = c(\epsilon)$ for $x \in K$.  We will fix some $\epsilon < \frac{\alpha}{2}$.

 Thus,

{\small{
\begin{eqnarray*}
\lefteqn{E \left[ U \left(\widetilde{X}^{n_m}_T \right) \right]} \\
&& \geq E \left[ 1_{\{\tau_m < T\}} U \left(\widehat{X}^{n_m}_T + \alpha \right) \right] + E \left[ 1_{\{ \tau_m = T \}}U \left( \widehat{X}^{(n_m,k_m)}_T \right) \right] \\
&& \geq E \left[ 1_{\{\tau_m < T\}} U' \left(\widehat{X}^{n_m}_T + \alpha \right) \cdot \alpha \right] + E \left[ 1_{\{\tau_m < T\}} U \left(\widehat{X}^{n_m}_T\right) \right] + E \left[ 1_{\{ \tau_m = T \}}U \left( \widehat{X}^{(n_m,k_m)}_T \right) \right].
\end{eqnarray*}
}}
By Corollary \ref{cora.3}, we have 
\[
E \left[ 1_{\{\tau_m < T\}} U \left(\widehat{X}^{n_m}_T\right) \right] + E \left[ 1_{\{ \tau_m = T \}}U \left( \widehat{X}^{(n_m,k_m)}_T \right) \right] \ra E \left[ U \left(\widehat{X}^{n_m}_T\right) \right],
\]
as $m \ra \infty$.

From \eqref{lasteq5}, we know that $U' \left(\widehat{X}^{n_m}_T + \alpha \right) \geq c$ up to a set of measure $\epsilon$.  We then have
\[
\underset{m \ra \infty}{\liminf} \ \left(E \left[ U \left(\widetilde{X}^{n_m}_T \right) \right] -  E \left[ U \left( \widehat{X}^{n_m}_T \right) \right] \right)\geq \left(\frac{\alpha}{2} - \epsilon \right)c \alpha > 0.
\]
This, however, contradicts the optimality of $\widehat{X}^{n_m}_T$ when $m$ is sufficiently large.
\end{proof}

Now we can prove the main result of this section.

\begin{proposition}\label{propa.1} Suppose that $u^n \ra u^\infty$.  Then for each $i>0$, the set $\{(\widehat{X}^{(n,i)})^* : n \in \mb{N} \}$ is bounded in probability.
\end{proposition}

\begin{proof}  It is true by construction that for each $k$, $\left \{ \underset{0 \leq t \leq T}{\inf} \widehat{X}^{(n,k)}_t : n \in \mb{N} \right \}$ is bounded in probability, since $\widehat{X}^{(n,k)} > - k$.  To conclude, it only remains to apply Lemma \ref{lemmaa.5}.

\end{proof}

\section{Uniformly Integrable Wealth Processes: A Brief Counterexample}\label{exp_app2}

The next proposition is based directly from an example of \cite{Schachermayer00howpotential}, which can be easily modified to fit the setting of this paper.  It addresses the second claim made in Remark \ref{lemma5remark}.

\begin{proposition}\label{propa.2} There exists a single market for which the optimal wealth process $(\widehat{X}_t)_{0 \leq t \leq T}$ does not have $\{\exp(-\widehat{X}_\tau) : \tau \in \mc{T} \}$ uniformly integrable.
\end{proposition}

\begin{proof}  Consider the example introduced on p. 13 of \cite{Schachermayer00howpotential}.  In that market, it is shown on p. 19 that the optimal wealth process $\widehat{X}$ satisfies $\underset{t \uparrow T}{\lim} \ E [-\exp(-\widehat{X}_t)] = - \infty$.  This clearly is not possible if  $\{\exp(-\widehat{X}_\tau) : \tau \in \mc{T} \}$ is uniformly integrable.
\end{proof}martingale measure and minimal entropy martingale measure %coincide.  Then $\underset{n \ra \infty}{\lim} \ v^n(y) = v^\infty(y)$.  Hence, $\underset{n \ra \infty}{\lim} \ u^n(x) = u^\infty(x)$.

\bibliographystyle{plain}
\bibliography{BMO2ref}
\end{document}